\documentclass[%
 reprint,
 amsmath,amssymb,
]{revtex4-2}
\usepackage{graphicx}
\usepackage{subcaption}
\usepackage{dcolumn}
\usepackage{bm}
\usepackage{xcolor}

\newtheorem{theorem}{Theorem}[section]
\newtheorem{corollary}[theorem]{Corollary}
\newtheorem{proposition}[theorem]{Proposition}

\newtheorem{lemma}[theorem]{Lemma}

\newenvironment{proof}[1][Proof]{\begin{trivlist}
\item[\hskip \labelsep {\bfseries #1}]}{\end{trivlist}}
\allowdisplaybreaks

\newenvironment{remark}[1][Remark]{\begin{trivlist}
\item[\hskip \labelsep {\bfseries #1}]}{\end{trivlist}}

\newcommand{\qed}{\nobreak \ifvmode \relax \else
	\ifdim\lastskip<1.5em \hskip- \lastskip
	\hskip 0.5em plus0em minus0.5em \fi \nobreak
	\vrule height0.75em width0.5em depth0.25em\fi}

\usepackage{hyperref}
\begin{document}

\preprint{APS/123-QED}

\title{Conformal symmetries and MOTS stability}

\author{Abbas M Sherif}
\email{abbasmsherif25@gmail.com}
 \altaffiliation{Institute of Mathematics, Henan Academy of Sciences (HNAS), 228 Mingli Road, Zhengzhou 450046, Henan, China.}

\date{\today}

\begin{abstract}
Let $\{\Sigma_t\}$ be a spacelike foliation of a spacetime $\mathcal{M}$, and suppose each $\Sigma_t$ is foliated by 2-surfaces $\mathcal{S}$ with spacelike unit normal in $\Sigma_t$. We show that under mild energy conditions, a MOTS $\mathcal{S}$ that intersects integral curves of past-pointing conformal Killing vector field lying in the normal space of $\mathcal{S}$ is strictly stable and evolves smoothly to a spacelike horizon. We also show that if the restriction of the divergence of the vector field to $\mathcal{S}$ is non-negative, $\mathcal{S}$ is unstable, and if negative and $\mathcal{S}$ is a 2-sphere, $\mathcal{S}$ must be strictly stable. 
\end{abstract}

\keywords{Black hole horizons; Conformal Symmetry; Marginally outer trapped surfaces (MOTS); MOTS stability}

\maketitle

\section{\label{sec:level1}Introduction}


Black holes have emerged as a key ``laboratory'' for merging quantum and gravity physics, with their utilities being exploited by both numerical and mathematical relativists. Their use in understanding thermal properties of spacetimes is well noted. While initially the event horizon provided the fundamental understanding of many properties of black holes, their teleological nature made them not suitable for the understanding of how these objects locally evolve.

The notion of ``trapped surfaces'' emerged in the seminal work of Penrose, \cite{pen1,pen2}, to formalize the description of a spacetime singularity. These are surfaces on which all null geodesics intersecting them orthogonally converge. Their presence, under some mild curvature conditions, leads to the formation of a singularity. (Also see \cite{sen1} for a comprehensive review of the singularity theorems.) Cosmic censorship \cite{pen3} posits that the formation of a horizon should precede the formation of a singularity to shield the singularity from the domain of outer communication.

When a spacetime is foliated by a family of spacelike hypersurfaces which are asymptotically Minkowski, one refers to a point $p$ in a hypersurface as trapped provided it is contained in a trapped surface in the hypersurface. An apparent horizon is the boundary of the collection of all of these trapped points. If there is a consistent way to choose an ``outward'' and ``inward'' to a surface, many authors prefer to define the marginally trapped condition as the vanishing of the expansion of the tangent to the outward pointing null congruences, and such a surface is called a marginally outer trapped surfaces (MOTS) \cite{hay0,hay1}. The apparent horizon will satisfy the MOTS condition. The two however do not exactly coincide, obviously.                                                       

A MOTS in general will evolve to foliate a marginally outer trapped tubes (MOTT) which is more general than the apparent horizon, and the particular character of the evolution depends on some additional constraints \cite{ash1,ash2,ib1,jar0}. The MOTS condition does away with the asymptotically Minkowski requirement in the definition of an apparent horizon, as well as the general impracticability of checking the boundary of all of such trapped points.

Of the several (albeit equivalent) approaches to understanding the local evolution of MOTS, i. e. whether or not a MOTS smoothly evolves into a MOTT, the formulation via a certain notion of stability introduced by Andersson \textit{et al.} in \cite{and1} has won out to be the most widely used in recent literature. (Also see \cite{and2} which expands on the results of \cite{and1}, and \cite{ib10} where the operator appears from computing normal variations of the expansion.). Here the stability of a MOTS is formulated as an eigenvalue problem for a certain linear second order elliptic operator. The sign of the principal eigenvalue of the operator determines if the MOTS is stable or unstable. Under suitable conditions, a strictly stable MOTS, where the principal eigenvalue is positive, will evolve to a smooth horizon, and if the shear of the MOTS along the outward direction is not identically zero or the null energy condition (NEC) strictly holds somewhere on the MOTS, this horizon will be spacelike.

In the presence of Killing symmetries, the existence/non-existence of MOTS and their evolution have been examined in various works. Mars and Senovilla \cite{mar0} have ruled out the existence of MOTS in a strictly stationary spacetime. Subsequent works (see for example \cite{mar2,mar3} and references therein) have proved valuable in elucidating the restrictions on the presence of MOT in spacetimes admitting conformal Killing vectors (CKV) and MOTS stability. At the initial data level, some incredible results have been obtained in \cite{ib5}, which have gone further in examining not just the principal eigenvalues of the stability operator but also the nature of the negative eigenvalues in the case of unstable MOTS. (A modest contribution to a part of the work of \cite{ib5} was provided in a note \cite{as3}, where a decomposition of the symmetry vector was assumed and the roles of the zeros of its components on stability/instability were analyzed.)

In this letter, we show that in the presence of a CKV with a particular orientation, and with the NEC assumption, any MOTS intersecting integral curves of the CKV will be strictly stable and hence will smoothly evolve to a horizon by Theorems 1 and 2 of Andersson {\em et al.}. In particular, we prove the following local evolution result for MOTS. 
\begin{theorem}\label{th75}
Let $\mathcal{M}$ be a 1+3 spacetime whose spacelike leaves are foliated by \(2\)-surfaces with unit normal $n^a$. Let $\mathcal{S}$ be one of these surfaces on which the NEC strictly holds somewhere, and suppose $\eta^a$ is a past-pointing CKV in the normal space of $\mathcal{S}$, which is a KV for the conformal metric in a neighborhood $\mathcal{U}\subseteq\mathcal{M}$. Then, if $\mathcal{S}$ is a MOTS in $\mathcal{U}$, it smoothly evolves to a spacelike horizon. 
\end{theorem}

By past pointing it is meant that the inner product with the canonical timelike vector field is positive. Subsequently, from the arguments involved in proving the above theorem, the following result follows.
\begin{theorem}\label{th76}
Let $\mathcal{M}$ be a 1+3 spacetime whose spacelike leaves are foliated by \(2\)-surfaces with unit normal $n^a$. Let $\mathcal{S}$ be one of these surfaces on which the NEC strictly holds somewhere, and suppose $\eta^a$ is a CKV in the normal space of $\mathcal{S}$, which is a KV for the conformal metric in a neighborhood $\mathcal{U}\subseteq\mathcal{M}$. Then, for a stable MOTS $\mathcal{S}$ in $\mathcal{U}$ and a point $p\in\mathcal{S}$, $\eta^a$ cannot lie to the interior of the light cone at $p$.
\end{theorem}

From Theorem \ref{th76} the following result immediately follows as a corollary.
\begin{corollary}\label{th77}
Let $\mathcal{M}$ be a 1+3 spacetime whose spacelike leaves are foliated by \(2\)-surfaces with unit normal $n^a$. Let $\mathcal{S}$ be one of these surfaces on which the NEC strictly holds somewhere, and suppose $\eta^a$ is a CKV in the normal space of $\mathcal{S}$, which is a KV for the conformal metric in a neighborhood $\mathcal{U}\subseteq\mathcal{M}$. Then, for any MOTS $\mathcal{S}$ in $\mathcal{U}$ and a point $p\in\mathcal{S}$, if $\eta^a$ timelike at $p$, $\mathcal{S}$ is unstable.
\end{corollary}

The condition that the CKV generates a KV for the conformal metric, crucial to our results, provides for the simplifications that go into our proofs. This requirement also establishes some properties that had previously been ruled out for a proper CKV. For example, under this assumption (even without imposing a preferred orientation on the CKV) it will be shown that for a local diffeomorphism $\Xi$ generated by the CKV, $\Xi(\mathcal{S})$ is a MOTT. In fact, this is true in the case of a homothetic Killing vector, including the Killing vector case (see for example the proof of Theorem 4 of \cite{mar3}). This follows from the fact that the variation of the expansion on a surface (whether it is a MOTS or not) is proportional to the expansion, so that in the case of a MOTS, the smooth family of surfaces are all MOTS. In the case of a CKV however, an extra term involving the variation of the divergence of the CKV appears. Under our assumptions, as will be seen, this term coincides with the expansion as well, up to a sign.

This paper is structured as follows. Section \ref{2} briefly introduces the geometry and dynamics of MOTS, including the notion of stability of MOTS. In Section \ref{3}, we introduce the relationship between properties of conformally related metrics on a spacetime and the MOTS condition, with the MOTS condition being a property of the variation of the divergence of the conformal Killing vector field along null geodesics. In Section \ref{4}, the stability of MOTS is analyzed and used to prove the main results of this work, Theorems \ref{th75} and \ref{th76}. We also comment on stability/instability of a MOTS established from properties of the conformal divergence. 


\section{The MOTS geometry and stability}\label{2}


We consider a spacetime $(\mathcal{M},g_{ab})$ foliated into spacelike slices $\Sigma_t$ of the spacetime, with unit timelike normal vector $u^a$ (we point the reader to the references \cite{hay1,ash1,ash2,ib1}). We pick a slice $\Sigma_0\in \{\Sigma_t\}$ and assume a foliation of $\Sigma_0$ by spacelike 2-surfaces with spacelike unit normal vector field $n^a$ in $\Sigma_0$ with $n_au^a=0$. For a 2-surface $\mathcal{S}$ in $\Sigma_0$, we fix the null gauge 
\begin{align}
k^a=u^a+n^a;\quad l^a=u^a-n^a,
\end{align}
with normalization $k_al^a=-2$, which spans the normal space at each point of $\mathcal{S}$. Then, the spacetime metric induces a 2-metric $q_{AB}$ on $\mathcal{S}$ as
\begin{align*}
q^{ab}=q^{AB}e_A^ae_B^b=g^{ab}+k^{(a}l^{b)},
\end{align*}
where $e^a_{A}$ are the pushfoward/pullback operators. These surfaces $\mathcal{S}$ are our surfaces of interest here. 

For almost all of the results that are to follow, it is emphasized that these surfaces of interest have unit normal in their embedding slice. This is to ensure that $n^a$ may be extended off a surface $\mathcal{S}$. However, it is important to point out that it is sufficient for the foliation to hold in some neighborhood of $\mathcal{S}$ for the subsequent results to be valid.

Now, the divergences of the congruences generated by \(k^a\) and \(l^a\) are respectively
\begin{align}
\mathcal{D}_ak^a&=\theta_k=z_1+z_2;\quad\mathcal{D}_al^a=\theta_l=z_1-z_2,\label{3.14}
\end{align}
where we have defined the scalars $z_1=\mathcal{D}_au^a$ and $ z_2=\mathcal{D}_an^a$, with $\mathcal{D}_a={q_a}^b\nabla_b$ being the compatible covariant derivative on the surface $\mathcal{S}$. The scalar $z_2$ is the mean curvature of the surface $\mathcal{S}$ and when it vanishes identically on $\mathcal{S}$, $\mathcal{S}$ is said to be a minimal surface and all standard results for minimal surfaces apply. A point $p\in\mathcal{S}$ will be called a minimal point of $\mathcal{S}$ if $z_2$ vanishes at $p$. 

We have fixed the orientation of $\mathcal{S}$ such that $k^a$ is outward pointing to $\mathcal{S}$ and $l^a$ is inward pointing. Then, respectively, the divergences are referred to as the {\em outgoing} and {\em ingoing} null expansions.

A 2-surface $\mathcal{S}$ is called a marginally outer trapped surface, abbreviated to MOTS, if the outgoing null expansion $\theta_k$ vanishes everywhere on $\mathcal{S}$. And if $\mathcal{S}$ is also a minimal surface, it will be referred to as a minimal MOTS. If the ingoing expansion is further constrained to satisfy $\theta_l<0$ everywhere on $\mathcal{S}$, the ``outer'' is dropped and $\mathcal{S}$ is simply called a MTS (or just a {\em marginal surface} as per the language of Hayward's \cite{hay1}). 

We note our convention here is such that we are choosing the definition for MOTS for which the outgoing expansion vanishes. Since the mean curvature vector along the MOTS is
\begin{align}
H^a\propto-\theta_lk^a-\theta_kl^a,
\end{align} 
our convention is requiring the mean curvature to align with the outgoing null direction in the case of a MOTS, i.e. it is future pointing. So for any trapping mentioned in this work we are referring to future trapping.

MOTS may be arranged to foliate a hypersurface, called a {\em marginally outer trapped tube} (MOTT) (this would be an MTT, sometimes referred to as holographic screen, if the surfaces are MTS), which can be spacelike, null, or timelike at a point. And if the MOTT is smooth we call it a horizon. In some cases for which a MOTT is everywhere null or spacelike, it will bound a black hole, with the latter requiring the null energy condition to strictly hold somewhere, in which case it is referred to as a dynamical horizon. 

In order to define the notion of stability, one considers a 1-parameter family of surfaces $\mathcal{S}_v$, obtained by deforming a MOTS $\mathcal{S}=\mathcal{S}_{v=0}$ along the flow generated by the vector field $\psi n^a$. Then, the stability operator captures how the outgoing null expansion $\theta_k$ associated to $\mathcal{S}$ changes under such deformation, which was shown to be equivalent to the action of a second order elliptic operator on the function $\psi$ \cite{and1,and2}:
\begin{align}
\delta_{\psi n}\theta_k=L_{\mathcal{S}}\psi,\label{st1}
\end{align}  
where the differential operator $L_{\mathcal{S}}$ takes the explicit form
\begin{align}
L_{\mathcal{S}}&=-\Delta_{\mathcal{S}}+s^a\mathcal{D}_a+Y\nonumber\\
Y&=\frac{1}{2}\left(\mathcal{R}+2\mathcal{D}_as^a-2s_as^a-\mathcal{G}_{ab}k^al^b\right).\label{st2}
\end{align}
Here, $\mathcal{R}$ is the scalar curvature of the MOTS, the 1-form $2s_a=-l_bq^c_a\nabla_ck^b$ is the connection on the normal bundle of the MOTS (sometimes referred to as the H\'{a}ji\u{c}ek 1-form and is a generator of rotation), and $\mathcal{G}_{ab}$ is the Einstein tensor. If on $\mathcal{S}$ there is a (and not identically zero) scalar $\psi$ such that $L_{\mathcal{S}}\psi>0$ ({\em resp.} $L_{\mathcal{S}}\psi=0$), $\mathcal{S}$ is characterized as strictly stable ({\em resp.} marginally stable). Otherwise, $\mathcal{S}$ is said to be unstable. The equivalence between the sign of $L_{\mathcal{S}}\psi$ and the principal eigenvalue $\lambda_0$ of $L_{\mathcal{S}}$, i.e. that with the smallest real part which is always real, was established in \cite{and1,and2}. That is, solving the eigenvalue problem
\begin{align}
L_{\mathcal{S}}\psi=\lambda\psi,\label{st3}
\end{align}  
for the principal eigenvalue $\lambda_0$, characterizes stability: $\lambda_0\geq0$ implies stability and $\lambda_0>0$ implies strict stability. The marginal case $\lambda_0=0$ provides a complicated picture where a vanishing eigenvalue leads to bifurcation and other interesting characters (see for example the reference \cite{boo13}).


\section{MOTS in spacetimes with conformally related metrics}\label{3}


Two metrics \(g_{ab}\) and \(\tilde{g}_{ab}\) on a spacetime \(\mathcal{M}\) are said to be conformally related if
\begin{eqnarray}\label{2.1}
\tilde{g}_{ab}=e^{2\varphi}g_{ab},
\end{eqnarray}
for some smooth function \(\varphi\) on \(\mathcal{M}\). Such a symmetry may be generated by a vector field which we notate as \(\eta^a\), called the conformal Killing vector (CKV), which satisfies the {\em conformal Killing equation} (CKE) equation
\begin{eqnarray}\label{2.2}
\mathcal{L}_{\eta}g_{ab}=2\nabla_{(a}\eta_{b)}=2\varphi g_{ab}.
\end{eqnarray}
Here $\mathcal{L}_{\eta}$ is the Lie derivative along the vector field $\eta^a$, $\nabla_a$ the covariant derivative with respect to the spacetime metric $g_{ab}$, $4\varphi=\nabla_a\eta^a$ is the conformal divergence, and the round brackets indicates symmetrization. For $\varphi=\mbox{constant}\neq0$, $\eta^a$ is a homothetic Killing vector (HKV), and for $\varphi=0$, $\eta^a$ is a Killing vector (KV). Otherwise, $\eta^a$ is a {\em proper} CKV. The metric $\tilde{g}_{ab}$ will be referred to as the {\em conformal metric}. And depending on the sign of the conformal divergence, we will say that the conformal observers (those along orbits of $\eta^a$) are diverging ($\varphi>0$) or converging ($\varphi<0$).

Let $\eta^a$ be a smooth vector field in the normal space to a surface $\mathcal{S}$. Then, as $\{k^a,l^a\}$ is a basis for the normal space, we may write
\begin{align}
\eta^a&=\bar{\eta}k^a+\tilde\eta l^a,\label{2.3}
\end{align} 
for functions $\bar{\eta}$ and $\tilde{\eta}$. It will further be assumed that $\eta^a$ is not identically zero on $\mathcal{S}$. Let us suppose that $\eta^a$ is a CKV for the spacetime with metric $g_{ab}$. Then, the Lie derivative of the conformal metric $\tilde{g}_{ab}$ along $\eta^a$ is
\begin{align}
\mathcal{L}_{\eta}\tilde{g}_{ab}=2\Psi\tilde{g}_{ab},\label{2.4}
\end{align}
where we have written
\begin{eqnarray}\label{2.5}
\Psi=\mathcal{L}_{\eta}\varphi+\varphi.
\end{eqnarray}
Here $\Psi$ is the conformal divergence with respect to the conformal metric. 

Now, if we project the conformal Killing equation (CKE) \eqref{2.2} to the surface $\mathcal{S}$ we get
\begin{align}
2\varphi=(\bar{\eta}+\tilde{\eta})z_1 + (\bar{\eta}-\tilde{\eta})z_2.\label{2.6}
\end{align}
Therefore, by choosing
\begin{align}
\dot{\varphi}=-\frac{1}{2}z_1,\quad \varphi'=-\frac{1}{2}z_2,\label{2.7}
\end{align}
where we have introduced the ``dot'' and ``prime'' notations for the Lie derivatives along the directions $u^a$ and $n^a$, respectively, we have $\mathcal{L}_{\eta}\varphi=-\varphi$, so that we have $\Psi=0$. That is, with this choice of the convective derivatives of the conformal divergence, $\eta^a$ is a KV for the conformal metric. From the above we can state the following result.
\begin{proposition}\label{th1}
Let $\mathcal{M}$ be a 1+3 spacetime whose spacelike leaves are foliated by \(2\)-surfaces with unit normal $n^a$. Let $\mathcal{S}$ be one of the surfaces. Any CKV in the normal space of $\mathcal{S}$ determines a KV for the conformal metric, in a neighborhood $\mathcal{U}\subseteq\mathcal{M}$, via \eqref{2.7}. 
\end{proposition}
It is obvious that if such a mapping exists, it has to be unique.

\begin{remark}[Remark 1.]
Notice that $\eta^a$ being an HKV would immediately imply that it is a KV by the choice \eqref{2.7} in which case we have that either $\eta^a$ generates an isometry (since $\varphi$ will vanish on any 2-surface in a neighborhood of the CKV) or $\eta^a$ is proper. 
\end{remark}
We clarify that Proposition \ref{th1} relies on whether or not the conditions \eqref{2.7} hold. So, whenever \eqref{2.7} fails, subsequent results in the work will generally not be valid. 

Now, it is straightforward to see the following relationship between the null expansion and the evolution of the conformal divergence
\begin{align}
2\mathcal{L}_k\varphi=-\theta_k.\label{var1}
\end{align}
Hence, if for a given surface $\mathcal{S}$ the conformal divergence $\varphi$ is, at all points of $\mathcal{S}$, constant along outgoing geodesics emanating from that surface, $\mathcal{S}$ is a MOTS. That is, under the assumptions of Proposition \ref{th1}, if \eqref{2.3} is a CKV on $\mathcal{M}$, which is a KV for the conformal metric, then for a given surface $\mathcal{S}$ in a leaf of $\mathcal{M}$, $\mathcal{S}$ is a MOTS if the conformal divergence is constant along the outgoing null geodesics emanating from $\mathcal{S}$. This may be seen as an existence result for MOTS.

If a 2-surface $\mathcal{S}$ is a MOTS, then by \eqref{2.6} we have the conformal divergence on $\mathcal{S}$ as
\begin{align}
\varphi|_{\mathcal{S}}=-\tilde\eta z_2.\label{2.6ex}
\end{align}
For any minimal surface $\mathcal{S}$, the MOTS condition requires that $\eta^a$ is at least locally a KV by \eqref{2.6}. One can see the ``at least locally'' from another perspective: away from $\mathcal{S}$, the spacetime gradient $\nabla_a\varphi=\mathcal{D}_a\varphi$ means that in general $\eta^a$ does not even have to be an HKV, except when we have $\varphi$ is constant on each 2-surface in the leaves of $\mathcal{M}$. That is, minimal MOTS are forbidden for a signed $\varphi$. This was discussed in some generality in \cite{mar0}.

On the other hand, the conformal divergence will have a constant sign at $p\in\mathcal{S}$ provided that the mean curvature of $\mathcal{S}$ is constant at $p$ and the components of $\eta^a$ are constant there. In particular, $\varphi|_{\mathcal{S}}$ will vanish at $p\in\mathcal{S}$ if and only if $p$ is both a minimal point of $\mathcal{S}$ and a null point for $\eta^a$. Otherwise, $\varphi|_{\mathcal{S}}$ is signed.

\textbf{Class of applicable spacetimes.} A class of spacetimes naturally admitting the kind of decomposition that is assumed in this work is the class of locally rotationally symmetrc spacetimes, containing many of the well known non-rotating black hole solutions in general relatvity. These spacetimes admit a multiply transitive isometry group, with a continuous isotropy group at each point, and are locally represented by the metric \cite{stew1}
\begin{align}
ds^2=-a^2dt^2+b^2dx^2+c^2(dy^2+f^2dz^2),\label{dsmet1}
\end{align}
where $a,b$ and $c$ are functions of $t$ and $x$ only, and $f$ is a function of $y$. In terms of the metric functions, the conditions \eqref{2.7} are 
\begin{align}
\varphi_t=-(\ln c)_t,\quad \varphi_x=-(\ln c)_x,
\end{align}
i.e. $\varphi=-(\ln c)+d$, for some constant $d$. In other words, the conformal divergence has to be related to the area of the surface for this to work. Whenever $c=c(x)$ (for example, when the slices are foliated by topological 2-spheres of radius $x$), of course $\varphi=\varphi(x)$, in which case, necessarily, we have that any MOTS in a region of $\mathcal{M}$ where \eqref{2.7} holds is minimal.

Moreover, the approach will work equally well in the case of perturbed LRS geometries: for a perturbed LRS geometry, once the $\{u,n\}$ frame is fixed, the spacetime may admit an ``{\em almost Killing vector field}'', for which a CKV is a very special case \cite{rich1,taub1}. The verification of \eqref{2.7} can then be made to apply the stability results that are to follow.

\textbf{Establishing integrability.} Given a CKV on a spacetime $\mathcal{M}$ admitting the required decomposition and a 2-surface $\mathcal{S}$, we may compute the conformal divergence of the ckv and then perform a simple check on the directional derivatives of the conformal divergence to verify \eqref{2.7}. However, we would like to obtain integrability condition for \eqref{2.7} to know when it holds in a region of the spacetime from information pertaining to the curvature of the surface. For simplicity, let us restrict to those cases where the gradient of the conformal divergence has no component on $\mathcal{S}$. 

Now, in order for $u^a$ and $n^a$ to be surface forming the Lie bracket of $u^a$ and $n^a$ acting on an arbitrary function $\psi$, has to satisfy
\begin{align}
\left[u,n\right]\psi=(-\dot{\psi}\dot{u}_a+\psi'u_a')n^a.\label{integ1}
\end{align}
Now, if we label the affine parameters of integral curves of $u^a$ and $n^a$ by $\tau$ and $\omega$, respectively, this would be equivalent to the commuting of the mixed partial derivatives of $\varphi$, i.e. That is, $\psi_{\tau\omega}=\psi_{\omega\tau}$. Replace $\psi$ by $\varphi$, and by substituting \eqref{2.7} into the left and right hand sides of \eqref{integ1} (now the explicit forms of $z_i$ are used) we have the integrability condition given by
\begin{align}
(z_1\dot{u}_a-z_2u_a')n^a=Z_{ab}\nabla^al^b+\tilde Z_{ab}\nabla^ak^b,\label{integ2}
\end{align}
where we have now defined
\begin{align*}
Z_{ab}=K_{ab}^{(k)}+q^c_{(a}\nabla_{b)}k_c,\quad \tilde Z_{ab}=K_{ab}^{(l)}+q^c_{(a}\nabla_{b)}l_c,
\end{align*}
with the tensor
\begin{align*}
K_{ab}^{(j)}=q^c_aq^d_b\nabla_cj_d
\end{align*}
denoting the extrinsic curvature of the surface along a null normal vector field $j^a$. Since the extrinsic curvature of the surface may be decomposed into its trace and trace-free parts as
\begin{align}
K_{ab}^{(j)}=\theta_jq_{ab}+\sigma_{ab}^{(j)},
\end{align}
with $\sigma_{ab}^{(j)}=(q^c_aq^d_b-(1/2)q_{ab}q^{cd})\nabla_cj_d$ denoting the $j$-shear and $\theta_j$ the 2-trace of $j^a$, the conditions of \eqref{integ1} can be written as
\begin{align}
(z_1\dot{u}_a-z_2u_a')n^a=2\theta_k\theta_l&+(\sigma_{ab}^{(k)}+q^c_{(a}\nabla_{b)}k_c)\nabla^al^b\nonumber\\
&+(\sigma_{ab}^{(l)}+q^c_{(a}\nabla_{b)}l_c)\nabla^ak^b.\label{integ4}
\end{align}
Note that the first term vanishes on a MOTS in a region where \eqref{2.7} holds.

It is clear from \eqref{integ1} that in static regions the integrability condition reduces to 
\begin{align*}
z_2(u_a'n^a)=0.
\end{align*} 
So, either the mean curvature vanishes where the CKV is a KV for the conformal metric, or the propagation of the temporal vector field will have no component along the unit spatial normal to the surface. The latter, in particular, in the case that \eqref{2.7} holds on the embedding slice, is equivalent to the slice being time-symmetric (the mean curvature of the slice vanishes). In the context of MOTS as considered throughout this work, a MOTS in a time-symmetric slice is necessarily minimal. Therefore, in either case in static regions containing a MOTS, if \eqref{2.7} holds, $\varphi$ vanishes, i.e. $\eta^a$ is a KV there, and hence the MOTS cannot be stable due to the results of Mars and Senovilla \cite{mar0}.

For specificity, we may consider the LRS class of spacetimes in \eqref{dsmet1}. In this case the integrability condition may be explicitly expressed in a very simple form:
\begin{align}
c_t(\ln a)_x+c_x(\ln b)_t=0.\label{integ3}
\end{align}
providing for a simple check on the metric functions. It is clear that this always holds in static regions.

The applicability of our results to de Sitter spacetime which is obviously LRS, analyzed throughout \cite{mar3}, may be ruled out. We first check directly the criteria \eqref{2.7}: consider the umbilical clicing of the de Sitter metric (we restrict ourselves to dimension four here) given by
\begin{align}
ds^2=\frac{1}{\bar\delta^2\cos^2\sigma}\left(-d\sigma^2+g_{\mathbb{S}^3}\right),\label{dsmet1}
\end{align}
with $\sigma\in(-\pi/2,\pi/2)$ increasing to the future and $g_{\mathbb{S}^3}$ is the metric on the round 3-sphere:
\begin{align}
g_{\mathbb{S}^3}=r^2(dR^2+\sin^2R(d\theta^2+\sin^2\theta d\phi^2)),
\end{align} 
and for which the spacetime Ricci curvature is $R_{ab}=3\bar\delta^2g_{ab}$, where $\bar\delta$ is a positive constant. The metric \eqref{dsmet1} admits the timelike and future-pointing vector field $\eta^a=(\bar\delta\cos\sigma)^{-1}u^a$. The temporal and spatial unit vector fields are respectively $u^a=-(\delta\cos\sigma)\partial/\partial\sigma$ and $n^a=(\bar\delta\cos\sigma/r)\partial/\partial R$. We now compute 
\begin{align}
z_1=-2\bar\delta\sin\sigma,\quad z_2=\frac{2}{r}\bar\delta\cos\sigma((\ln r)_R+\cot R),
\end{align}
so that the conformal divergence, using \eqref{2.6}, computes to
\begin{align}
\varphi=\bar\eta z_1=-\tan\sigma.
\end{align}
One then checks that \eqref{2.7} is equivalent to the requirement 
\begin{align}
z_2=0\quad\mbox{or}\quad \sec\sigma=\sin\sigma,
\end{align}
both of which fail. 

We may also check through the integrability condition \eqref{integ3} which again reduces to the vanishing of $z_2$ since the first term vanishes and $c_R\neq0$.


\section{Stability results}\label{4}


In order to prove the main results of this work, we begin by establishing the following Lemma.
\begin{lemma}
Let $\mathcal{M}$ be a 1+3 spacetime satisfying the NEC whose spacelike leaves are foliated by \(2\)-surfaces with unit normal $n^a$. Let $\mathcal{S}$ be one of these surfaces and suppose $\eta^a$ is a CKV in the normal space of $\mathcal{S}$, which is a KV for the conformal metric in a neighborhood $\mathcal{U}\subseteq\mathcal{M}$. Then, for a MOTS $\mathcal{S}$ in $\mathcal{U}$, the variation of the expansion along the CKV vanishes identically on $\mathcal{S}$.
\end{lemma}
\begin{proof}
In the case that the variation is taken along the conformal symmetry vector, one has, for the expansion on the surface $\mathcal{S}$ (the general case where the vector field is not necessarily a CKV is obtained in \cite{mar2}, but we adapt their result to the case of a CKV for the purpose of the current work),
\begin{align}
\delta_{\eta}\theta_k&=2\mathcal{L}_k\varphi.\label{3.1}
\end{align}
And with the choice \eqref{2.7}, following from the relationship between the null expansion and the variation of the conformal divergence, 
\begin{align}
\delta_{\eta}\theta_k&=-\theta_k.\label{3.2}
\end{align}
Thus, the variation is identically zero if $\mathcal{S}$ is a MOTS.\qed
\end{proof}

Note that if the conformal divergence is a temporal function, i.e. the gradient of $\varphi$ is timelike and past-pointing, then, of course on an $\mathcal{S}$, $\mathcal{L}_k\varphi>0$. So, a 2-surface $\mathcal{S}$ in a neighborhood where the CKV is defined, with \eqref{2.7} valid, cannot be a MOTS. On the other hand, the relation \eqref{2.7} cannot hold for $\mathcal{L}_k\varphi\neq0$. 

The vanishing of the variation implies that each $\mathcal{S}_v'$, for $v>0$, in the family of surfaces $\{\mathcal{S}_v\}$, is a MOTS, along integral lines of $\eta^a$. That is, under our current considerations, a MOTS in a given leaf of $\mathcal{M}$ will propagate to nearby leaves and foliate a MOTT. This can be stated as the following MOTT existence result. 
\begin{proposition}\label{mottex}
Let $\mathcal{M}$ be a 1+3 spacetime, satisfying the NEC whose spacelike leaves are foliated by \(2\)-surfaces with unit normal $n^a$. Let $\mathcal{S}$ be one of these surfaces and suppose $\eta^a$ is a CKV in the normal space of $\mathcal{S}$, which is a KV for the conformal metric in a neighborhood $\mathcal{U}\subseteq\mathcal{M}$. If $\mathcal{S}$ is a MOTS in $\mathcal{U}$, then for a local diffeomorphism $\Xi$ generated by $\eta^a$, $\Xi(\mathcal{S})$ is a MOTT.
\end{proposition}

Additionally, if one imposes the null energy condition, it can be shown to imply that for small $v$, each $\mathcal{S}_v'$ is weakly outer trapped. Then, one may evoke the same arguments used in proving Theorem 4 of \cite{mar3} to establish that $\tilde\eta$ can be nowhere negative on $\mathcal{S}$. As will later be seen, this non-negativity is quite crucial for stability. 

Let us now recall the following result.
\begin{proposition}[Proposition 1, Mars {\em et al.} \cite{mar3}]\label{marprop}
Let $\mathcal{S}$ be a MOTS in an $n\geq3$-dimensional spacetime $(\mathcal{M},g_{ab})$ satisfying the NEC. Assume that there exists a future causal vector field $\xi^a$ along $\mathcal{S}$ which is not everywhere proportional to $k^a$ and such that $\delta_{\xi}\theta_k\geq0$. Then $\mathcal{S}$ cannot be strictly stable. Moreover, if $\delta_{\xi}\theta_k$ is positive somewhere then $\mathcal{S}$ cannot be marginally stable either.
\end{proposition}

We draw on the above result for the following reason. To allow for strict stability, Proposition \ref{marprop} requires that
\begin{itemize}
\item[(i)] $\eta^a$ is past causal; or

\item[(ii)] Future (or past) spacelike,
\end{itemize}
although marginal stability is allowed for.

We are now in the position to prove our main results.

\subsection*{Proof of Theorem \ref{th75}} 

To prove this result we will show that if $\mathcal{S}$ is a MOTS, then under the assumptions of the theorem, it is strictly stable. The conclusion of its smooth evolution then follows from Theorems 1 and 2 of \cite{and1}. 

By defining 
\begin{align*}
W=K_{ab}^{(k)}K^{ab(k)}+\mathcal{G}_{ab}k^ak^b,
\end{align*}
it was shown in \cite{and2} that
\begin{align}
L_{\mathcal{S}}\tilde\eta=-\frac{1}{2}(\delta_{\eta}\theta_k+4\bar\eta W),\label{carop1}
\end{align}
which, for our current consideration, due to the vanishing of the variation, becomes
\begin{align}
L_{\mathcal{S}}\tilde\eta=-2\bar\eta W,\label{carop1}
\end{align}

We look for a variation for which $L_{\mathcal{S}}\tilde\eta$ coincides with the usual stability operator. If we are to define the scaled normal
\begin{align}
\underline{\eta}^a=-(2\tilde\eta)n^a=-\tilde\eta(k^a-l^a),\label{carop2}
\end{align}
then the operator in \eqref{carop1} coincides with the usual MOTS stability operator \eqref{st2}. That is,
\begin{align}
L_{\mathcal{S}}\tilde\eta=L_{\mathcal{S}}\psi.\label{carop3}
\end{align}
See, for example, Remark 2 of \cite{mar3} (one may also consult the Eq. (2.23) of the reference \cite{ib10} where the relationship is quite transparent). In particular, $\tilde\eta$ will be the associated eigenfunction of the principal eigenvalue, i.e.  
\begin{align}
\lambda_0\tilde\eta=-2\bar\eta W.
\end{align} 
Then, we may take $\tilde\eta$ to be everywhere positive and the stability of the MOTS is determined by the sign configurations of the scalars $W$ and $\bar\eta$. If the NEC strictly holds somewhere on $\mathcal{S}$, then $W|_{\mathcal{S}}>0$,  and stability of $\mathcal{S}$ is characterized as follows:
\begin{align*}
\bar\eta|_{\mathcal{S}}&\leq0,\quad\mbox{(stable)}\\
\bar\eta|_{\mathcal{S}}&<0,\quad\mbox{(strictly stable)}\\
\bar\eta|_{\mathcal{S}}&>0.\quad\mbox{(unstable)}
\end{align*}
And because $\eta^a$ is past-pointing, $\tilde\eta>0$ implies $\bar\eta<0$ and hence $\mathcal{S}$ is strictly stable, which concludes the proof.\qed\\

If $\eta^a$ is future-pointing and causal, then the MOTS is either unstable or marginally stable by Proposition \ref{marprop}, and in the marginally stable case the MOTS evolves to a null MOTT . Conversely, marginal stabilty is only compatible with a future-pointing (but not necessarily a causal) $\eta^a$. \\

\subsection*{Proof of Theorem \ref{th76}} 

For a stable MOTS, the CKV, depending on whether or not the CKV is past or future-pointing, may or may not lie in the light cone. Now, because the MOTS is stable, as before $\tilde\eta$ is chosen strictly positive on the MOTS. In particular, the CKV cannot lie in the interior of the light cone at $p$, i.e. $\eta^a$ is either spacelike or null (of course the null case is ruled out in the case of a past-pointing $\eta^a$). Necessarily, the CKV points inward to $\mathcal{S}$.\qed\\

In the marginally stable case $\eta^a$ will lie on the future light cone, pointing inward to $\mathcal{S}$. The allowed region for the CKV in the case of a stable MOTS has been sketched in Fig. \ref{cone1}.
 \begin{figure}[ht!]
 \centering
  \includegraphics[width=0.2\textwidth]{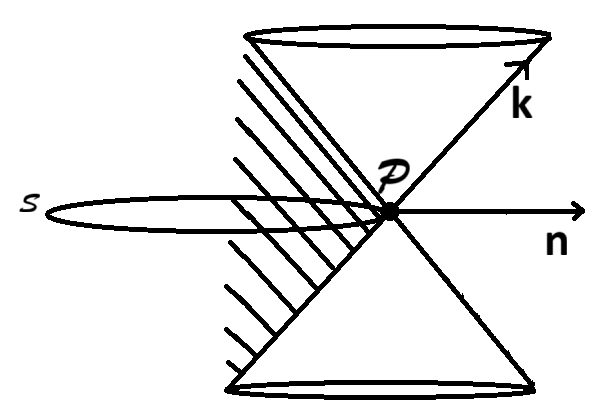}
 \caption{The light cone through a point $p$ on a stable MOTS $\mathcal{S}$: A depiction of the allowed region for $\eta^a$, above the outgoing null line.}
\label{cone1}
 \end{figure}

In the case that the CKV $\eta^a$ is future or past-pointing, it will lie above or below the spacelike normal, respectively, as sketched in Fig. \ref{cone2}. If future-pointing, then it must lie on the future light cone inward to $\mathcal{S}$.
\begin{figure}[ht!]
	\centering
	\begin{subfigure}{0.4\linewidth}
		\includegraphics[width=\linewidth]{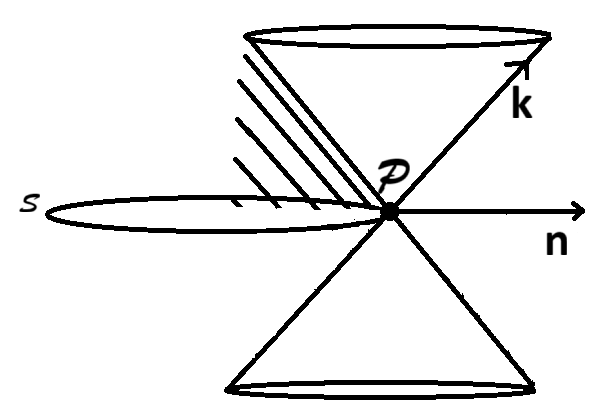}
		\caption{Allowed region for a future pointing $\eta^a$.}
		\label{fig:subfigA}
	\end{subfigure}
\hfill
	\begin{subfigure}{0.4\linewidth}
		\includegraphics[width=\linewidth]{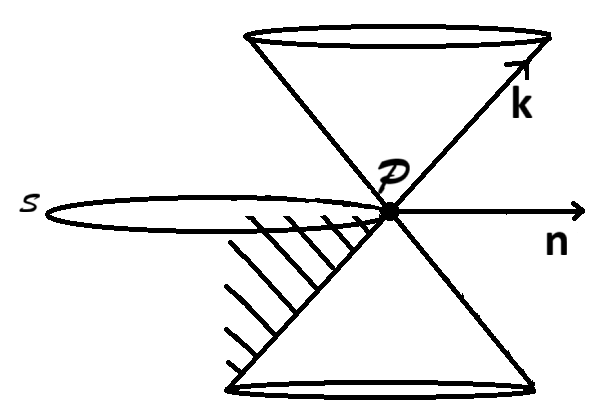}
		\caption{Allowed region for a past pointing $\eta^a$.}
		\label{fig:subfigB}
	\end{subfigure}
	\caption{Light cone through a point $p$ on a stable MOTS $\mathcal{S}$: A depiction of the allowed regions for a future-pointing $\eta^a$ (a) and a past-pointing $\eta^a$ (b).}
	\label{cone2}
\end{figure}

Now, in the case of a stable MOTS we know that $\tilde\eta>0$ on $\mathcal{S}$. We also know that aspherical MOTS are generically unstable, and so to consider stable MOTS we consider $z_2>0$. Then, $\varphi|_{\mathcal{S}}<0$ on $\mathcal{S}$, by \eqref{2.6ex}. This then is a necessary condition for strict stability. In particular, it follows that 
\begin{theorem}\label{th7}
Let $\mathcal{M}$ be a 1+3 spacetime whose spacelike leaves are foliated by \(2\)-surfaces with unit normal $n^a$. Let $\mathcal{S}$ be one of these surfaces on which the NEC strictly holds somewhere, and suppose $\eta^a$ is a CKV in the normal space of and not everywhere null on $\mathcal{S}$, which is a KV for the conformal metric in a neighborhood $\mathcal{U}\subseteq\mathcal{M}$. Then, on any MOTS $\mathcal{S}$ in $\mathcal{U}$, if $\varphi|_{\mathcal{S}}\geq0$, $\mathcal{S}$ cannot be stable. 
\end{theorem}

Conversely, for a spherical MOTS $\mathcal{S}$ such that $\varphi|_{\mathcal{S}}<0$ in a spacetime for which the NEC holds, $\tilde\eta>0$ on $\mathcal{S}$. And if $\eta^a$ is past-pointing, $\bar\eta<0$. This therefore leads to the following.
\begin{theorem}\label{thx7}
Let $\mathcal{M}$ be a 1+3 spacetime, satisfying the NEC whose spacelike leaves are foliated by \(2\)-surfaces with unit normal $n^a$. Let $\mathcal{S}$ be one of these surfaces on which the NEC strictly holds somewhere, and suppose $\eta^a$ is a past-pointing CKV in the normal space of $\mathcal{S}$, which is a KV for the conformal metric in a neighborhood $\mathcal{U}\subseteq\mathcal{M}$. Then, any spherical MOTS $\mathcal{S}$ in $\mathcal{U}$ on which $\varphi|_{\mathcal{S}}<0$ is strictly stable and evolves to a smooth spacelike horizon.
\end{theorem}


\begin{acknowledgments}


The author is thankful to the anonymous reviewer for very helpful comments. The author acknowledges that this research is supported by the Institute of Mathematics, funded through the High-level Talent Research Start-up Project Funding of the Henan Academy of Sciences (Project No.: 251819085).
\end{acknowledgments}

\end{document}